\documentclass[copyright]{eptcs}
\providecommand{\event}{J.~Dassow, G.~Pighizzini, B.~Truthe (Eds.): 11th International Workshop\\
on Descriptional Complexity of Formal Systems (DCFS 2009)} 
\providecommand{\volume}{3}
\providecommand{\anno}{2009}
\providecommand{\firstpage}{97} 
\providecommand{\eid}{9}
\usepackage{breakurl}        

\input{dcfs.tex}

\newcommand{\bdisplay}{\begin{description}\footnotesize\item[]}
\newcommand{\edisplay}{\end{description}}

\newcommand{\bquot}[1]{\begin{quotation}\small\noindent
  \textbf{#1}\hspace{\labelsep}\ignorespaces}
\newcommand{\equot}{\unskip\end{quotation}}
\usepackage{amsfonts}

\newcommand{\ra}{\rightarrow}
\newcommand{\Ra}{\Rightarrow}
\newcommand{\PC}{\mathit{PC}}
\newcommand{\NPC}{\mathit{NPC}}
\newcommand{\CF}{\mathit{CF}}
\newcommand{\RE}{\mathit{RE}}

\begin{document}
     
\title{On the Size Complexity of  Non-Returning Context-Free\\ PC Grammar Systems}

\def\titlerunning{Size complexity of non-returning context-free PC grammar systems}

\author{Erzs\'ebet Csuhaj-Varj\'u
\institute{Computer and Automation Research Institute --
  Hungarian Academy of Sciences \\
  Kende u. 13--17 -- 1111 Budapest -- Hungary}
\institute{E\"otv\"os Lor\'and University --
Faculty of Informatics -- 
Department of Algorithms and Their Applications\\
P\'azm\'any P\'eter s\'et\'any 1/c -- 1117 Budapest -- Hungary}
\email{csuhaj@sztaki.hu}
\and
Gy\"orgy Vaszil
\institute{Computer and Automation Research Institute --
  Hungarian Academy of Sciences \\
  Kende u. 13--17 -- 1111 Budapest -- Hungary}
\email{vaszil@sztaki.hu}
}
\def\authorrunning{E.~Csuhaj-Varj\'u, Gy.~Vaszil}
               
\maketitle

\begin{abstract}
Improving the previously known best bound, we show that any
recursively enumerable language can be generated with a non-return\-ing
parallel communicating (PC) grammar system having six context-free
components. We also present a non-returning universal PC grammar
system generating unary languages, that is, a system where not only
the number of components, but also the number of productions and the
number of nonterminals are limited by certain constants, and
these size parameters do not depend on the generated language.    
\end{abstract}

\section{Introduction}

Parallel communicating grammar systems (PC grammar systems, for short)
are network architectures for distributed generation of languages
\cite{PaSa89}.  In these systems, the component grammars generate
their own sentential forms in parallel, and their activity is organized
in a communicating system.  Two basic variants of PC grammar systems
are distinguished: In so-called returning systems, after communication,
the component starts a new derivation (``returns'' to its axiom), while
in so-called non-returning systems it continues the rewriting of its
current sentential form. The language generated by a PC grammar system
is the set of terminal words generated by a distinguished component
grammar called the master.

An important problem regarding parallel communicating grammar systems
is how much succinct descriptions of languages they provide: For
example, what is the minimal number of components, nonterminals,
and/or productions that generating PC grammar systems (or its
individual components) need to obtain a language in a certain language
class.  Especially interesting question is, if for a fixed language
class some of these parameters can be bounded by suitable constants,
how many of them can be limited at the same time.

During the years, a considerable amount of research was devoted to the
examination of the power and the size of PC grammar systems with
context-free components (context-free PC grammar systems), but the
question whether or not these constructs are computationally complete
was open for a long time. (For some basic results, consult
\cite{CsuDaKePa,hbgs}).

Obtained independently from each other, it was shown that both
returning \cite{CsuVa} and non-returning context-free PC grammar
systems \cite{Ma} are able to generate any recursively enumerable
language.  Since non-returning systems can be simulated with returning
systems, the second result implies the first one, but in \cite{Ma} no
bound was given on the number of components, while the construction 
used in~\cite{CsuVa} provided~11 as an upper bound.  In \cite{CsuPaVa} this
number was decreased to 5, the best known bound so far.  To give an
upper bound on the necessary number of components of non-returning
context-free PC grammar systems which are able to generate any
recursively enumerable language, a construction simulating a
two-counter machine with a non-returning context-free PC grammar
system with 8 components was presented in \cite{Va}.

The fact that a bounded number of components is enough to generate any
recursively enumerable language inspired further investigations of the
size complexity of returning context-free PC grammar systems. In
\cite{CsuVa2} a trade-off between the number of rules or nonterminals
and the number of components is demonstrated: With no bound on the
number of components, 7 rules and 8 nonterminals in each of the
component
grammars are sufficient to generate any recursively enumerable
language, while if the number of rules and nonterminals can be
arbitrary high, then the number of components can be bounded by a
constant.

In this paper, we continue the above line of investigations.  As an
improvement of the previous bound, we show that non-returning PC
grammar systems with 6 context-free components are computationally
complete, i.\,e., they are able to determine any recursively enumerable
language. Furthermore, based on the results in \cite{Kor96}, where
universal register machines with a number of rules limited by small
constants are provided, we present constant bounds on the
size complexity parameters of a so-called non-returning universal PC
grammar system generating unary languages.

\section{Preliminaries and definitions}

The reader is assumed to be familiar with the basic notions of formal
language theory; for further information we refer to \cite{hb}.  The
set of non-empty words over an alphabet $V$ is denoted by $V^+$; if the
empty word, $\lambda$, is included, then we use the notation $V^*.$ A
set of words $L\subseteq V^*$ is called a language over $V.$ For a
word $w\in V^*$ and a set of symbols $A\subseteq V$, we denote the
length of $w$ by $|w|$, and the number of occurrences of symbols from
$A$ in $w$ by $|w|_A$. If $A$ is a singleton set, $A=\{a\}$, then we
omit the brackets and write $|w|_a$ instead of $|w|_{\{a\}}$.  The
families of context-free languages and recursively enumerable
languages are denoted by ${\mathcal L}(\CF)$ and ${\mathcal L}(\RE)$.

A {\it two-counter machine}, see \cite{Fi66}, 
$M=(\Sigma\cup\{Z,B\},E,R,q_0,q_F)$
is a 3-tape Turing machine where~$\Sigma$ is an {\it alphabet}, $E$ is a
set of {\it internal states} with two distinct elements $q_0,q_F\in
E$, and $R$ is a set of {\it transition rules}.  The machine has a
read-only input tape and two semi-infinite storage tapes (the
counters). The alphabet of the storage tapes contains only two
symbols, $Z$ and $B$ (blank), while the alphabet of the input tape is
$\Sigma\cup \{B\}$.  
The symbol $Z$ is written on the first, leftmost cells of the storage
tapes which are scanned initially by the storage tape heads, and may
never appear on any other cell. An integer $t$ can be stored by moving
a tape head $t$ cells to the right of $Z$. A stored number can be
incremented or decremented by moving the tape head right or left. The
machine is capable of checking whether a stored value is $zero$ or not
by looking at the symbol scanned by the storage tape heads.  If the
scanned symbol is $Z$, then the value stored in the corresponding
counter is $zero$ (which cannot be decremented since the tape head
cannot be moved to the left of $Z$).

The rule set $R$ contains transition rules of the form
$(q,x,c_1,c_2)\!\rightarrow\!(q',e_1,e_2)$ where 
\hbox{$x\!\in\! \Sigma\!\cup\!\{B\}\!\cup\!\{\lambda\}$}
corresponds to the symbol scanned on the input tape in state $q\in E$,
and $c_1,c_2\in\{Z,B\}$ correspond to the symbols scanned on the
storage tapes.  
By a rule of the above form, $M$ enters state $q'\in E$, and the
counters are modified according to $e_1,e_2\in \{-1,0,+1\}$.  If $x\in
\Sigma\cup\{B\}$, then the machine was scanning~$x$ on the input tape, and the head
moves one cell to the right; if $x=\lambda$, then the machine
performs the transition irrespective of the scanned input symbol, and
the reading head does not move.

A word $w\in \Sigma^*$ is accepted by the two-counter machine if
starting in the initial state $q_0$, the input head reaches and reads
the rightmost non-blank symbol on the input tape, and the machine is
in the accepting state $q_F$.  Two-counter machines are
computationally complete; they are just as powerful as
Turing machines. 

Now we recall the definitions concerning parallel communicating
grammar systems (see \cite{PaSa89}); for more information we refer to
\cite{CsuDaKePa,hbgs}.

A {\it parallel communicating grammar system} with $n$ context-free
components is an $(n+3)$-tuple 
$$\Gamma =(N,K,\Sigma,G_1,\ldots, G_n),\ n\geq 1,$$
where $N$ is a {\it nonterminal alphabet}, $\Sigma$ is a
{\it terminal alphabet}, and $K=\{Q_1,\ldots, Q_n\}$ is an alphabet of
{\it query symbols}. The sets $N$, $\Sigma$, and $K$ are pairwise 
disjoint; $G_i=(N\cup K,\Sigma,P_i,S_i)$, $1\leq i \leq n$, called a {\it
component} of $\Gamma$, is a usual Chomsky grammar with the
nonterminal alphabet $N\cup K$, terminal alphabet~$\Sigma$, set of
rewriting rules $P_i\subset N\times (N\cup K\cup \Sigma)^*$, and {\it axiom}
(or start symbol) $S_i\in N$.
One of the components, $G_i$, is
distinguished and called the {\it master grammar} (or the master) of
$\Gamma.$


An $n$-tuple $(x_1,\ldots, x_n)$, where $x_i\in (N\cup \Sigma\cup
K)^*$, for $1\leq i \leq n$, is called a {\it configuration}
of~$\Gamma$;  $(S_1,\ldots,S_n)$ is said to be the {\it initial
configuration}.
%
PC grammar systems change their configurations by performing
direct derivation steps.
%
We say that $(x_1,\ldots, x_n)$ {\it directly derives} $(y_1,\ldots,
y_n),$ denoted by $(x_1,\ldots, x_n)\Rightarrow (y_1,\ldots,
y_n)$, if one of the following two cases holds:

1. There is no $x_i$ which contains any query symbol, that is, $x_i\in
(N\cup \Sigma)^*$ for all $1\leq i\leq n.$ Then, for each $i,$ $1\leq
i\leq n$, $x_i\Rightarrow_{G_i} y_i$ ($y_i$ is obtained from $x_i$ by
a direct derivation step in $G_i$) for $x_i\notin \Sigma^*$ and
$x_i=y_i$ for $x_i\in \Sigma^*$.

2. There is some $x_i,\ 1\leq i\leq n,$ which contains at least
one occurrence of a query symbol.
For each such $x_i,$ $1\leq i\leq n,$ with $|x_i|_K\ne 0$ we write
$x_i=z_1Q_{i_1}z_2Q_{i_2}\ldots z_tQ_{i_t}z_{t+1}$, where 
$z_j\in (N\cup \Sigma)^*$, \hbox{$1\leq j\leq t+1$}, and $Q_{i_l}\in K$, $1\leq l \leq t.$
If $|x_{i_l}|_K=0$ for each $l,$ $1\leq l\leq t,$ then
$y_i=z_1x_{i_1}z_2x_{i_2} \ldots z_tx_{i_t}z_{t+1}$ and (a) in {\it
returning} systems we have $y_{i_l}=S_{i_l}$, while (b) in {\it
non-returning} systems we have $y_{i_l}=x_{i_l}$, \hbox{$1\leq l\leq t$}. If
$|x_{i_l}|_K\ne 0$ for some $l,$ $1\leq l\leq t,$ then $y_i=x_i.$ For
all~$j$, $1\leq j \leq n,$ for which $y_j$ is not specified above,
$y_j=x_j$.

Let $\Rightarrow^*$ denote the reflexive and transitive closure of
$\Rightarrow$.  Let the language generated by the component $G_i$
be denoted by $L(G_i)$, that is, 
\begin{eqnarray*}
L(G_i)&=&\{x\in \Sigma^*\mid (S_1,\dots,
S_i,\dots,S_n) \Rightarrow^*(x_1,\dots,x_i,\dots, x_n)\mbox{ for}
\\ 
&&\ \ \mbox{some }
x_1,\dots,x_n\in(N\cup \Sigma\cup K)^* \mbox{ such that }  x=x_i\}.
\end{eqnarray*} 
Then, the
{\it language generated} by the system $\Gamma$ is $L(\Gamma)=L(G_j)$
where $G_j$,  \hbox{$1\leq j\leq n$}, is the master component of the system.

Let the class of languages generated by returning and non-returning
PC grammar systems
having at most $n$ context-free components, where $n\ge 1,$ be
denoted by ${\mathcal   L}(\PC_n\CF)$ and ${\mathcal   L}(\NPC_n\CF)$,
respectively, and let 
${\cal L}(X_*\CF)=\bigcup\limits_{i\geq 1}{\cal L}(X_i\CF),$ 
$X\in \{\PC, \NPC\}$.

Using these notations, the results on the generative power of
context-free PC grammar systems can be summarized as follows
(for details, see \cite{CsuDaKePa,CsuPaVa,hbgs,Ma,Va}):
$$
{\cal L}(\CF)\subset{\cal L}(X_2\CF)\subseteq{\cal L}(\PC_5\CF)={\cal L}(\PC_*\CF)=
{\cal L}(\NPC_8\CF)={\cal L}(\NPC_*\CF)={\cal L}(\RE),
$$
for $X\in\{\PC,\NPC\}$.

\section{Improving the bound on the number of components}

In the following we show that every recursively enumerable language
can be generated by a non-returning PC grammar system with six context-free
components.

\begin{theorem}
$
{\cal L}(\NPC_6\CF)={\cal L}(\RE).
$
\end{theorem}

\begin{proof}
Let $L\subseteq \Sigma^*$ be an arbitrary recursively enumerable
language and $M=(\Sigma\cup\{Z,B\},E,R,q_0,q_F)$ be a
two-counter machine accepting $L$. Without the loss of the generality
we may assume that $M$ always enters the final state with
empty counters and lets them unchanged, i.\,e.,  for any $q\in E$ with
$(q,x,c_1,c_2)\rightarrow (q_F,e_1,e_2)\in R$ it holds that
$c_1=c_2=Z$ and $e_1=e_2=0.$

To prove the statement, we construct a non-returning context-free PC grammar system
$\Gamma$ generating~$L$.  Let
$
\Gamma=(N,K,\Sigma,G_{sel},G_{gen},G_{c_{1}},G_{c_{2}},G_{ch_1},
G_{ch_2}),
$ 
where $G_{gen}$ is the master grammar and
$G_\gamma=(N,K,\Sigma,P_\gamma,\omega_\gamma)$ is a component grammar
for $\gamma\in\{gen,sel, {c_{1}},{c_{2}},{ch_1},{ch_2}\}$ and
$\omega_\gamma$ is the axiom.

Let $\mathcal I=\{[q,x,c_1,c_2,q',e_1,e_2]\mid
(q,x,c_1,c_2)\rightarrow(q',e_1,e_2)\in R\}$ and let us introduce for
any\linebreak \hbox{$\alpha=[q,x,c_1,c_2,q',e_1,e_2]\in \mathcal{I}$} the following
notations: $State(\alpha)=q$, $Read(\alpha)=x$, $NextState(\alpha
)=q'$, and $Store(\alpha,i)=c_i$, $Action(\alpha,i)=e_i,$ where
$i=1,2$.



The simulation is based on representing the states and the transitions
of $M$ with nonterminals from~$\mathcal{I}$ and the values
of the counters by strings of nonterminals containing as many symbols
$A$ as the value stored in the given counter. Every component is
dedicated to simulating a certain type of activity of the two-counter
machine: $G_{sel}$ selects the transition to be simulated, $G_{c_i},$
where $1\leq i\leq 2,$ simulates the respective counter and the
update of its contents, $G_{ch_j}$,
where $1\leq j\leq 2,$ assists the work of $G_{c_i},$ and~$G_{gen}$
generates the word read (and possibly accepted) by $M.$

Let
$N=\mathcal{I}\cup\{S,A,Z,F,F',F'',F''',C_{1},C_2, 
M_0,M_1,M_2\}\cup   
\{D_{i,\alpha},E_{i,\alpha}, H_{i,\alpha}\mid \alpha \in \mathcal{I}, 
1\leq i\leq 2\}$
and let the axioms and the rules of the components be defined as
follows.
Let $\omega_{sel}= S$, 
\begin{eqnarray*}
P_{sel}&=& \{S\rightarrow \alpha \mid \alpha \in {\mathcal I}, 
State(\alpha)=q_0\}\cup {}
\{\alpha\ra  D_{1,\alpha}, D_{1,\alpha}\ra D_{2,\alpha}\mid
\alpha \in\mathcal{I}\}\cup {}
\\
&& \{D_{2,\alpha}\rightarrow \beta \mid  \alpha, \beta \in \mathcal{I}, 
\ NextState(\alpha)=State(\beta) \}\cup {}
\\
& & \{D_{2,\alpha}\rightarrow {F}\mid  \alpha \in {\mathcal I},
\ NextState(\alpha)=q_F\}\cup {}
 \{F\rightarrow F\}.
\end{eqnarray*}

This component selects the transition of the two-counter machine to be
simulated. The axiom $S$ is used to initialize the system by
introducing one of the symbols from $\mathcal I$ denoting an initial 
transition, i.\,e., a
symbol of the form $[q_0,x,c_1,c_2,q',e_1,e_2]$ where $q_0$ is the
initial state. The other productions are used for changing the
transition into the next one to be performed. The appearance of symbol
$F$ indicates that the simulation of the last transition has been
finished and the rule $F\rightarrow F$ can be used to 
continue rewriting until the other components also finish their work.
 Let $\omega_{gen}=S,$
\begin{eqnarray*}
P_{gen}&=&\{S\rightarrow Q_{sel},C_{1}\rightarrow C_{2},  
C_{2}\rightarrow  Q_{sel},
F\rightarrow F', F'\rightarrow Q_{ch_1}Q_{c_1}Q_{c_2}\}\cup 
\\
& & \{\alpha \rightarrow xC_{1}\mid  \alpha \in \mathcal{I},
Read(\alpha)=x \}\cup 
\\
&& \{H_{2,\alpha}\rightarrow \lambda\mid \alpha \in \mathcal{I}\}
\cup \{M_1\rightarrow \lambda ,Z\rightarrow \lambda, 
F''\rightarrow \lambda, F'''\rightarrow \lambda\}.
\end{eqnarray*}

This component generates the string accepted by the counter machine by
adding the symbol\linebreak \hbox{$x=Read(\alpha)$} for each $\alpha\in {\mathcal I}$
(chosen by the selector component $G_{sel}$) using the rule
$\alpha\ra xC_1$. The productions rewriting $C_1$ to $C_2$ and then
$C_2$ to $Q_{sel}$
are used for maintaining the synchronization.  The result of
the computation is produced by using rules 
$F\rightarrow F', F'\rightarrow Q_{ch_1}Q_{c_1}Q_{c_2}$.
 After the symbol $F$ appears, the component makes sure that the
strings obtained from components $G_{c_1},$ $G_{c_2}$ and $G_{ch_1}$
do not contain any nonterminal letter which is different from 
$H_{2,\alpha}$, for $\alpha \in \mathcal{I}$, or from any of
$M_1,Z, F'',F'''$, since these are the only symbols which can be 
erased. (The symbols $H_{2,\alpha}$, for $\alpha \in
\mathcal{I}$, and $M_1$ indicate that the simulation of the checks and
the updates of the contents of the counters of the two-counter machine
were correct; $Z$ is an auxiliary symbol;
$F''$ and  $F'''$ are different variants
of the symbol denoting the final transition.) 
If the work of the component stops with a  terminal word, then this string was also accepted by $M$ and 
the simulation was correct.

The following two components are for representing the contents of the
counters of $M$ and for simulating the changes in the stored
values.
Let for $i\in\{1,2\}$, $\omega_{c_{i}}=S,$
\begin{multline*}
P_{c_{i}} = \{S\rightarrow Q_{sel}Z,A\rightarrow Q_{ch_2},
F\rightarrow F'', F''\ra F''\}\cup {}
\\
\qquad\:\:\:\: \{ \alpha \rightarrow Q_{sel}, 
D_{2,\alpha}\rightarrow Q_{sel}y_{i,\alpha} \mid 
\alpha \in {\mathcal I},Store(\alpha,i)\!=\!B, 
y_{i,\alpha}\!=\!\sigma(Action(\alpha,i),Store(\alpha,i))\}\cup {}
\\ 
\hspace*{-26mm} \{\alpha \rightarrow H_{1,\alpha}, 
H_{1,\alpha}\rightarrow H_{2,\alpha},
H_{2,\alpha}\rightarrow Q_{sel}y_{i,\alpha}\mid \alpha\in{\mathcal I}, 
\ Store(\alpha,i)=Z,\\
 y_{i,\alpha}=\sigma(Action(\alpha,i), 
Store(\alpha,i))\}
\end{multline*}
where $\sigma:\{1,0,-1\}\times\{B,Z\} \to \{AA,A,\lambda\}$ is a 
partial mapping defined as $\sigma(1,B)=AA,$ $\sigma(0,B)=A,$ 
$\sigma(-1,B)=\lambda,$
$\sigma(1,Z)=A,$ $\sigma(0,Z)=\lambda$. 

These components are responsible for simulating the change in the
contents of the  counters, which is represented by a string $u$
consisting of as many letters $A$ as the actual stored number in the
counter. By performing rule $A\rightarrow Q_{ch_2}$ and the rules 
$\alpha \rightarrow Q_{sel}, 
D_{2,\alpha}\rightarrow Q_{sel}y_{i,\alpha}$, the
components check whether the string representing the counter contents
contains at least one occurrence
of the letter $A$ (which is required by the transition represented by
$\alpha$), and then modify the contents of the counter in the
prescribed manner by introducing the necessary number of new $A$s 
contained in the string $y_{i,\alpha}$. 
If $Store(\alpha,i)=B$, then the simulation is correct if and only 
if one occurrence of $A$ is rewritten first, and then
productions $\alpha \rightarrow Q_{sel}, 
D_{2,\alpha}\rightarrow Q_{sel}y_{i,\alpha}$
are applied in the given order, i.\,e.,  after
three steps the new string will contain one occurrence of
$M_1$. Any other order of rule application results in introducing
either a letter for which no rule exists ($D_{1,\alpha}$ if $u$ has no
occurrence of $A$) or a letter which cannot be erased from the
sentential form anymore ($M_2$, if $A$ is rewritten in the
second step). 

If $Store(\alpha,i)=Z$, then the rules 
$\alpha \rightarrow H_{1,\alpha}, 
H_{1,\alpha}\rightarrow H_{2,\alpha}$, and $
H_{2,\alpha}\rightarrow Q_{sel}y_{i,\alpha}$
are used for checking whether $u$ 
contains an $A$. The required condition
holds and the simulation is successful if after applying the
productions, $H_{2,\alpha}$ appears in the second step in the 
new sentential form and it has no
occurrence of the symbol $A.$ The non-occurrence of $A$ will be checked
later by
components $G_{ch_1}$ and $G_{gen}$.
Let $\omega_{ch_1}=S,$
\begin{eqnarray*}
P_{ch_{1}} &=& \{S\rightarrow Q_{sel},\alpha \rightarrow E_{1,\alpha},  
E_{2,\alpha}\rightarrow Q_{sel}\}\cup {}
\\
&& \{E_{1,\alpha}\rightarrow E_{2,\alpha},
 \mid  \alpha \in {\mathcal I},
Store(\alpha,1)=B, Store(\alpha,2)=B\} \cup {}
\\
&& \{  
E_{1,\alpha}\rightarrow E_{2,\alpha}Q_{c_2}
 \mid \alpha \in {\mathcal I},
Store(\alpha,1)=B, Store(\alpha,2)=Z\} \cup {}
\\
&& \{  
E_{1,\alpha}\rightarrow E_{2,\alpha}Q_{c_1}
 \mid \alpha \in {\mathcal I},
Store(\alpha,1)=Z, Store(\alpha,2)=B\} \cup {}
\\
&& \{
E_{1,\alpha}\rightarrow E_{2,\alpha}Q_{c_1}Q_{c_2}
\mid \alpha \in {\mathcal I},
Store(\alpha,1)=Z,  Store(\alpha,2)=Z\} \cup {} 
\\
&&\{F\rightarrow F''',F'''\ra F'''\}.
\end{eqnarray*}

This component assists in checking whether the contents of the respective
counter is zero if it is required by the transition to be performed. 
This is done by asking the
string of the component $G_{c_1}$ and/or~$G_{c_2}$ after
the second step of the corresponding derivation phase. If the string
(or strings) communicated to this component contains (contain) an
occurrence of $A,$ then this letter will never be removed from the
sentential from since $P_{ch_1}$ has no rule for deleting $A$ and the
component $G_{gen}$ which will later issue a query to $G_{ch_1}$, 
has no erasing rule for $A$ either. This means that
the simulation is correct if the string or strings communicated to
$G_{ch_1}$ are free from $A$ but contains (contain) an occurrence 
of~$H_{2,\alpha}.$

Finally, let $\omega_{ch_2}\!=\!S$ and
$P_{ch_2}\! =\! \{S\!\rightarrow\! M_0, M_0\!\rightarrow\! M_{1},  
M_{1}\!\rightarrow\!  M_{2}, M_2\!\rightarrow\! M_0\}.$
This component assists~$G_{c_1}$ and $G_{c_2}$
in checking whether or not the string
representing the counter contents contains an occurrence of $A$.
The simulated counter is not
empty and the simulation is correct if and only if $P_{ch_2}$ is
queried in a step when the symbol $M_1$ is
communicated to the respective component $G_{c_1}$ or $G_{c_2}.$

In the following we  discuss the work of $\Gamma$ in details.
After the first rewriting step, we obtain a configuration
$
(S,S,S,S,S,S)\Ra
(\alpha_0,Q_{sel},Q_{sel}Z,Q_{sel}Z,Q_{sel},M_0)\Ra
(\alpha_0,\alpha_0,\alpha_0Z,\alpha_0Z,\alpha_0,M_0)
$
where $\alpha_0$ is a nonterminal denoting one of the initial
transitions of the two-counter machine, i.\,e., $State(\alpha_0)=q_0.$
Notice that since the two counters are empty at the beginning, the
sentential forms of components $G_{c_1}$ and $G_{c_2}$ do not contain
any occurrence of $A.$

In the following we demonstrate how the simulation works. We
consider a particular case, the proof of all other cases can be done
similarly.

Let $\alpha=[q,x,B,Z,q',e_1,e_2]\in \mathcal{I}$, where
$x\in \Sigma\cup\{\lambda\}$, $q,q'\in E$, and we do not specify
$e_1,e_2$ at this moment.  Furthermore, let $\beta \in
\mathcal{I}$ with $NextState(\alpha)=State(\beta).$
Suppose that up to transition $\alpha$ the simulation was
correct. Then the configuration of $\Gamma$ is of the form
$
(\alpha, w\alpha,\alpha uZ,\alpha vZ,\alpha \bar w,M_0)
$
where $w\in T^*$, $u,v\in \{A, M_1\}^*,$ and
$\bar w \in (\{M_1,Z,\}\cup \{H_{2,\alpha} \mid 
\alpha  \in \mathcal{I}\})^*$.

By the next rewriting step, $\alpha$ at the first component changes
into $D_{1,\alpha}$, and then by the second rewriting step into
$D_{2,\alpha}$. Similarly, $w\alpha$ changes into $wxC_1,$ and then
into $wxC_2$ where $x=Read(\alpha)$.

Let us examine now $\alpha uZ$ which represents the contents of the
first counter. Since, by the requirements of the simulated transition, 
the counter must not be
empty, $u$ should have at least one occurrence of $A$. If this is not
the case, then the only rule which can be applied is $\alpha
\rightarrow Q_{sel},$ which introduces $D_{1,\alpha}$ in the
string. Then the derivation gets blocked since there is no rule for
rewriting $D_{1,\alpha}$ or $Z$, thus the derivation cannot be
continued. 

If we suppose that $u$ has
at least one occurrence of $A$, then after two rewriting steps and
the communication following them, the following cases may hold: 
The new string contains $M_1$ and $D_{2,\alpha}$ (first an occurrence of $A$ and 
then $\alpha$ was rewritten), or it contains $M_1$ and $M_2$ (two
occurrences of $A$ were replaced), or it contains $D_{1,\alpha}$ and 
$M_2$ (first $\alpha$, then one occurrence of $A$ was rewritten). The two
latter cases do not lead to termination (and thus, correct simulation)
since neither $M_2$ nor $D_{1,\alpha}$ can be removed from the
string when it is later sent to the master component $G_{gen}$. (Unlike
$M_1$ and $D_{2,\alpha}$ which can be erased by $G_{gen}$.)

Therefore, after one more rewriting step, we must have a string of the
form $Q_{sel}y_1u_1M_1u_2Z$ where $u=u_1Au_2$
and $y_{1}$ corresponds to $e_1=Action(\alpha,1)$ 
for $\alpha=[q,x,B,Z,q',e_1,e_2]$ as follows:
Since one $A$ was removed from $u$, if $e_1=-1$ then 
$y_{1}=\lambda$,  if $e_1=0$ then 
$y_{1}=A$, and  if $e_1=+1$ then 
$y_{1}=AA$.

Let us consider now $\alpha vZ,$ i.\,e.,  the string
representing the contents of the second counter.  In this case $v$
must not have an appearance of $A$ (according to the current
transition
symbol $\alpha=[q,x,B,Z,q',e_1,e_2]$). If this is the case,
that is, if $|v|_A=0$, then the
only rule which can be applied is $\alpha\rightarrow H_{1,\alpha},$
and then the derivation continues with applying
$H_{1,\alpha}\rightarrow H_{2,\alpha}.$ After the second rewriting
step the new string will be of the form $H_{2,\alpha}vZ$ which will
be forwarded by request to component $G_{ch_1}$ and stored there until
the end of the derivation when it is sent to the master component 
$G_{gen}$. The grammar
$G_{gen}$ is not able to erase the nonterminal $A$, thus, 
terminal words can only be generated if $G_{ch_1}$ received 
a string representing the empty counter.

If we assume that $v$ contains at
least one copy of $A$, then after two rewriting steps we obtain
a string which has occurrences of either $M_1$ and $M_2$ (two copies of
$A$ were replaced), or $M_1$ and $H_{1,\alpha}$, or $H_{1,\alpha}$ and $M_2$ (in both cases one copy of $A$ was
rewritten), or $H_{2,\alpha}$ and $A$ (no copy of $A$ was rewritten,
but $|v|_A\not=0$.) None of
these cases can lead to a correct simulation, since as we have seen
above, these strings are transferred to $G_{ch_1}$ and then to $G_{gen}$
in a later phase of the derivation, where $M_2,$
$H_{1,\alpha}$, and $A$ cannot be deleted.

This means that the new string obtained from $\alpha
vZ$ after the third rewriting step must be of the form $Q_{sel}
y_2vZ,$ where $v$ contains no occurrence of $A$ and $y_2$
is the string corresponding to $e_2=Action(\alpha,2)$. Since,
in the case of a correct simulation, no $A$ was deleted,
$y_2=\lambda$ if $e_2=0$, and $y_2=A$ if $e_2=+1$ (the case
$e_2=-1$ is not applicable, since the counter is empty,
$Store(\alpha,2)=Z$).


Continuing the derivation, the prescribed communication step 
results in the configuration
$$
(\beta,wx\beta,\beta u'Z,\beta v'Z, \beta{\bar w}', M_0)
$$
where $\beta\in\mathcal I$ is a transition with $NextState(\alpha)=
State(\beta)$, $u',v'$ are strings representing the counters
of $M$ following the transition described by $\alpha\in \mathcal I$,
and $\bar w'$ is a string over \hbox{$\{M_1,Z\}\cup \{H_{2,\alpha}\mid
\alpha\in \mathcal I\}$}. Thus, we obtain a configuration of the form 
we started from. Now, similarly as above, the simulation of the 
transition corresponding to the symbol $\beta\in\mathcal I$ can
be performed.

Suppose now that $NextState(\alpha)=q_F$ and $G_{sel}$ decides to end
the simulation of $M$, that is, instead of $\beta$, the nonterminal
$D_{2,\alpha}$ is changed to $F$. Then the obtained configuration is 
$$
(F,wxF,Fu'Z,Fv'Z,F{\bar w}', M_0).
$$ 
Since $M$ always enters the final state with 
empty counters, we have $|u'|_A=|v'|_A=0$, thus we obtain
{\small$$
(F,wxF',F''u'Z,F''v'Z,F'''{\bar w}', M_1)\Ra
(F,wxQ_{ch_1}Q_{c_1}Q_{c_2},F''u'Z,F''v'Z,F'''{\bar w}',M_2),
$$}
and then
$
(F,wxF'''\bar w'F''u'ZF''v'Z,F''u'Z,F''v'Z,F'''{\bar w}', M_0).
$
We also know that in
case of a correct simulation, $|\bar w'|_A=0$, therefore
by applying the erasing rules of $P_{gen}$ 
to delete $H_{2,\alpha}$, $M_1$,
$Z$, $F''$, and~$F'''$, we either obtain a terminal word $w'=wx$ also
accepted by the two-counter machine $M$, or  there are nonterminals in the sentential form
of $G_{gen}$ which cannot be deleted. By the explanations above,
it can also be seen that $\Gamma$ generates the same language as $M$ accepts.
\end{proof}

\section{A universal PC grammar system for unary languages}

In the following we study the possibility of generating all
recursively enumerable languages (over a certain alphabet) 
with not only a bounded number of
components, but also with bounded measures of other kind, such as the
number of rewriting rules, or the number of nonterminals.  To this aim
we examine the possibility of simulating universal variants of Turing
machines.

Instead of universal two-counter
machines, we  consider the similar notion of register machines 
since several examples of very simple, but
still universal machines of this kind are known. Since register machines 
work with sets of non-negative integers, we also restrict ourselves to
the study of generating unary languages.

A register machine consists of a given number of registers
and a set of labeled instructions.
There are several types of instructions
which can be used:
\begin{itemize}
\item  
$l_i: ({\tt  ADD}(r), l_j)$ -- add 1 to  register $r$
and then go to the instruction with label $l_j$,

\item
$l_i:({\tt CHECK}(r), l_j,l_k)$ -- if the value of register $r$ 
is zero, go to instruction 
$l_j$, otherwise go to $l_k$,
\item 
$l_i: ({\tt CHECKSUB}(r), l_{j}, l_{k})$ -- if the value of
register $r$ is positive, 
then subtract~1 from it and go to the instruction with
label $l_{j}$, otherwise go to the instruction with label $l_{k}$,
\end{itemize}
and instruction $l_{h}: {\tt HALT}$ to halt the machine.
Thus, formally, a {\it register machine} is a construct\linebreak
\hbox{$M=(m, H, l_0,l_h,R)$}, where $m$ is the number of registers, $H$ is the set of
instruction labels, $l_0$ is the start label, $l_h$ is the halting
label, and $R$ is the set of instructions; each label from $H$ labels
exactly one instruction from $R$. A register machine $M$ computes a value
$y\in \mathbb N$ on input $x\in \mathbb N$
in the following way: it starts with the input $x$ in its input 
register by executing the instruction with
label $l_0$ and proceeds by applying instructions as indicated by the
labels. If the
halt instruction is reached, then the number $y\in \mathbb N$ 
stored at that time in
the output register is the result of the computation of $M$. If 
the machine does not halt, the result is undefined. 
It is known (see, e.\,g., \cite{Minsky}) that register machines  compute
the class of partial recursive functions.


Register machines with $n$ registers can also be simulated by
the straightforward generalization of
two-counter machines having $n$ counter tapes instead of two.
We call this model an {\it $n$-counter machine} in the following.
Given a register machine $M_1$ with $n$ registers, we
%
can easily construct an $n$-counter machine $M_2$ over
a unary input alphabet which simulates its computations. If the $n$
counter tapes of $M_2$ correspond to the $n$ registers of $M_1$,
and if
$M_2$ is started
with a unary input word $w$ and a value $x\in \mathbb N$ stored
on one of its  counter tapes (the one corresponding to the input 
register), then it can check whether $|w|=y\in \mathbb N$
is computed by $M_1$ on input $x$ by simulating the labeled 
instructions of the register machine.
To do this, the states of $M_2$ should correspond to the labels
of the instructions of $M_1$ and its transition relation should
be defined as follows.

To simulate an instruction $l_j: ({\tt  ADD}(r), l_k)$,
$M_2$ should have transition rules 
$$(l_j,\lambda,c_1,\ldots,c_n)\ra(l_k,e_1,\ldots,e_n)$$
for all possible combinations
of $c_i\in\{Z,B\}$, \hbox{$1\leq i\leq n$} and with $e_r=+1$, and $e_i=0$
for all \hbox{$1\leq i\leq n$}, $i\not= r$.

To simulate an instruction
$l_j:({\tt CHECK}(r), l_k,l_l)$, 
$M_2$ should have transition rules 
$$(l_j,\lambda,c_1,\ldots,c_n)\ra(l_k,0,\ldots,0)$$
for all combinations of $c_i\in\{Z,B\}$ where
$c_r=Z$, and also the transitions $(l_j,\lambda,c_1,\ldots,c_n)\ra
(l_l,0,\ldots,0)$ for all combinations of $c_i\in\{Z,B\}$ where
$c_r=B$, $1\leq i\leq n$.

An instruction $l_j: ({\tt CHECKSUB}(r), l_{k}, l_{l})$
can be simulated by similar transition rules if we replace
the ``don't change'' instruction corresponding to the $r$th counter
with ``subtract one'', that is, we replace the 
$0$ on the $(r+1)$th position on the right side of the transition
rule with $-1$.

The transitions of the counter machine $M_2$ defined above simulate the 
work of $M_1$ in the sense that whenever the state $l_h$ corresponding
to the halting instruction is reached after starting the
machine with $x\in \mathbb N$ stored on the input counter tape, then 
the value stored on the output counter tape, $y\in\mathbb N$, is the 
same as computed by the register machine $M_1$ on input $x$.
If we assume that the first counter corresponds to the output register
of $M_1$,
then to check whether the input word is of the form $w=a^y$, we need, for
all combinations of $c_i\in\{Z,B\},\ 2\leq i\leq n$, the 
transitions
$(l_h,a,B,c_2,\ldots,c_n)\ra (l_h,-1,0,\ldots,0)$ and 
$(l_h,\lambda,Z,c_2,\ldots,c_n)\ra (q_F,0,0,\ldots,0)$ where
$q_F$ is the final state of $M_2$.

In \cite{Kor96} several small universal register machines are presented.
One of them, which we call $U$ in the following, has eight registers and
it can simulate the computation of any register machine $M$ with the
help of a ``program'', an integer $code(M)\in \mathbb N$ coding the
particular machine $M$. If $code(M)$ is placed in the second register
and an argument $x\in \mathbb N$ is placed in the third register,
then~$U$  simulates the computation of $M$ by halting if
and only if $M$ halts, and by producing the same result in its first
register as $M$ produces in its output register after a halting
computation. Moreover, $U$ has eight {\tt ADD} instructions,
one {\tt CHECK} instruction, and twelve {\tt CHECKSUB}
instructions.

Based on the universal machine $U$ and the simulation
technique described above, we can obtain PC grammar
systems which are universal in the sense that they are able to
generate all languages over a certain fixed alphabet if
we initialize one of the components with a ``program'' corresponding 
to the language we wish to generate, that is, if the component
is started with  an axiom which
is a word different from the start symbol.

\begin{definition}\rm
A PC grammar system
$\Gamma=(N,K,T,G_1,\ldots,G_n)$ is {\em universal},
if there exists an index $j$, $1\leq j\leq n$,
such that for all languages $L\subseteq \Sigma^*$ over a finite alphabet
$\Sigma$, there is a word $w_L\in N^*$ with
\hbox{$L=L(\Gamma,w_L)=L(G_i,w_L,j)$} where
$$L(G_i,w_L,j)=\{x_i\in\Sigma^*\mid (\alpha_1,\ldots,\alpha_n)\Ra^*
(x_1,\ldots,x_n) \mbox{ for }\alpha_j=w_L,
 \alpha_i=S_i,\ 1\leq i\leq n,\ i\not=j\},$$
and $G_i$ is the master component of the system.
\end{definition}

Now based on the PC grammar system described in the previous section,
 we can obtain the following theorem.

\begin{theorem}
There exists a non-returning universal PC grammar system $\Gamma_U$,
such that
any recursively enumerable language $L$ over the unary alphabet
can be generated by $\Gamma_U$ as $L=L(\Gamma_U,w_L)$ for some word 
$w_L$ corresponding to $L$. 

Moreover, $\Gamma_U$ has at most 
12 components,  $48m+51$ rewriting rules, 
and $4m+12$ nonterminal symbols, where $m=23\cdot 2^8+3$.
\end{theorem}

\begin{proof}
The statement can be proved based on the discussions above.
Consider the universal register machine $U$ from \cite{Kor96},
having 8 registers and 21 instructions. 
We can construct an
8-counter machine $M_U$ which simulates the work of $U$ in the
sense described above, that is, if $M_U$ is started with the code of 
a register machine $M$ stored on its second counter tape and
an input $x\in \mathbb N$ stored on its third counter tape, then
it accepts the unary word $w$ written on its input tape if and only
if $|w|=y$, where $y\in\mathbb N$ is the value computed by $M$
on the input $x$.

$U$ has eight registers and, as we have explained above, we need a
different transition rule for the simulation of a given instruction for
each possible combination of empty and non-empty registers. This means
that we need $2^8$ transition rules for simulating each register
machine instruction, thus, we need~\hbox{$21\cdot 2^8$} rules 
to simulate the 21 instructions of $U$, and
$2^8$ additional rules for comparing the result (appearing on the
first counter tape) with the contents of
the input tape.

If we add a new starting state $q_0$, and the transitions
$(q_0,\lambda,Z,B,Z,\ldots,Z)\ra (q_0,0,0,+1,0,\ldots,0)$,
$(q_0,\lambda,Z,B,B,Z,\ldots,Z)\ra (q_0,0,0,+1,0,\ldots,0)$, and 
$(q_0,\lambda,Z,B,B,Z,\ldots,Z)\ra (l_0,0\ldots,0)$, thus,
we nondeterministically ``fill'' the input counter (corresponding to
the third counter tape) before starting the actual computation, then
we can obtain the possible results without placing any input in the 
third counter.
This means that we can accept any word $w$ with $|w|=y$
where $y\in\mathbb N$ is a value from the range of the function
computed by the register machine $M$. Thus, 
choosing the appropriate $M$, we can accept
the words of any recursively enumerable language over the unary 
alphabet by initializing only the second counter tape with the code
of the given machine $M$. 

If we also make sure that before entering the final state, the
contents of all the counters of the machine~$M_U$ are erased, then we
will be able to use a similar construction as in the proof of Theorem
1 to construct a non-returning PC grammar system $\Gamma_U$ for the
simulation of $M_U$. To erase the counter contents, we need $2^8$
transitions in addition, thus, altogether the counter machine $M_U$
has $m=23\cdot 2^8+3$ transition rules.

The PC grammar system that we obtain after applying the construction
based on the proof of Theorem 1 will be a universal system if instead
of the start symbol~$S$, we initialize the component $G_{c_2}$ 
corresponding to the second counter of $M_U$ with a word of the
form $A^nS$ where $n=code(M)$, such that the range of the function
computed by the register
machine $M$ corresponds to the length set of the words of the unary
language $L$. 

By observing the modified construction, the resulting system has $8+4=12$ components, $48\cdot
m+51$ rewriting rules, and $4\cdot m+12$ nonterminals, thus, we obtain
the bounds given in the statement of the theorem.
\end{proof}

\section{Conclusions}

We have improved the previously known bound on the number of
non-returning components necessary to generate any recursively
enumerable language.  We also presented a technique for the simulation
of register machines, and we used it to simulate a concrete example of
a small universal register machine.  We obtained a non-returning
universal PC grammar system which is able to generate any unary
recursively enumerable language.  Since the construction we used is
general, not taking advantage of any of the special properties of the
universal register machine that was simulated, it is expected that
with more precise observations, the rough bounds we have given above
can be further decreased.  We also propose to employ similar
techniques for the study of the descriptional complexity measures of
returning PC grammar systems.

\bibliographystyle{eptcs}
\bibliography{vaszil}

\begin{thebibliography}{10}
\providecommand{\bibitemstart}[1]{\bibitem{#1}}
\providecommand{\bibitemend}{}
\providecommand{\bibliographystart}{}
\providecommand{\bibliographyend}{}
\providecommand{\url}[1]{\texttt{#1}}
\providecommand{\urlprefix}{Available at }
\providecommand{\bibinfo}[2]{#2}
\bibliographystart

\bibitemstart{CsuDaKePa}
\bibinfo{author}{E.~Csuhaj-Varj\' u}, \bibinfo{author}{J.~Dassow},
  \bibinfo{author}{J.~Kelemen} \& \bibinfo{author}{Gh.~P\u aun}
  (\bibinfo{year}{1994}): \emph{\bibinfo{title}{Grammar Systems. A Grammatical
  Approach to Distribution and Cooperation}}.
\newblock \bibinfo{publisher}{Gordon and Breach, London}.
\bibitemend

\bibitemstart{CsuPaVa}
\bibinfo{author}{E.~Csuhaj-Varj\'u}, \bibinfo{author}{Gh.~P\u aun} \&
  \bibinfo{author}{Gy. Vaszil} (\bibinfo{year}{2003}): \emph{\bibinfo{title}{PC
  grammar systems with five context-free components generate all recursively
  enumerable languages}}.
\newblock {\sl \bibinfo{journal}{Theoretical Computer Science}}
  \bibinfo{volume}{299}, pp. \bibinfo{pages}{785--794}.
\bibitemend

\bibitemstart{CsuVa}
\bibinfo{author}{E.~Csuhaj-Varj\'u} \& \bibinfo{author}{Gy. Vaszil}
  (\bibinfo{year}{1999}): \emph{\bibinfo{title}{On the computational
  completeness of context-free parallel communicating grammar systems}}.
\newblock {\sl \bibinfo{journal}{Theoretical Computer Science}}
  \bibinfo{volume}{215}, pp. \bibinfo{pages}{349--358}.
\bibitemend

\bibitemstart{CsuVa2}
\bibinfo{author}{E.~Csuhaj-Varj\'u} \& \bibinfo{author}{Gy. Vaszil}
  (\bibinfo{year}{2002}): \emph{\bibinfo{title}{Parallel communicating grammar
  systems with bounded resources}}.
\newblock {\sl \bibinfo{journal}{Theoretical Computer Science}}
  \bibinfo{volume}{276}, pp. \bibinfo{pages}{205--219}.
\bibitemend

\bibitemstart{hbgs}
\bibinfo{author}{J.~Dassow}, \bibinfo{author}{Gh.~P\u aun} \&
  \bibinfo{author}{G.~Rozenberg} (\bibinfo{year}{1997}):
  \emph{\bibinfo{title}{Grammar systems}}.
\newblock In: \bibinfo{editor}{A.~Salomaa G.~Rozenberg}, editor: {\sl
  \bibinfo{booktitle}{Handbook of Formal Languages}}.
  \bibinfo{publisher}{Springer-Verlag, Berlin}, pp. \bibinfo{pages}{155--213}.
\bibitemend

\bibitemstart{Fi66}
\bibinfo{author}{P.~C. Fischer} (\bibinfo{year}{1966}):
  \emph{\bibinfo{title}{Turing machines with restricted memory access}}.
\newblock {\sl \bibinfo{journal}{Inform. and Control}} \bibinfo{volume}{9}, pp.
  \bibinfo{pages}{364--379}.
\bibitemend

\bibitemstart{Kor96}
\bibinfo{author}{I.~Korec} (\bibinfo{year}{1996}): \emph{\bibinfo{title}{Small
  universal register machines}}.
\newblock {\sl \bibinfo{journal}{Theoretical Computer Science}}
  \bibinfo{volume}{168}, pp. \bibinfo{pages}{267--301}.
\bibitemend

\bibitemstart{Ma}
\bibinfo{author}{N.~Mandache} (\bibinfo{year}{2000}): \emph{\bibinfo{title}{On
  the computational power of context-free PC grammar systems}}.
\newblock {\sl \bibinfo{journal}{Theoretical Computer Science}}
  \bibinfo{volume}{237}, pp. \bibinfo{pages}{135--148}.
\bibitemend

\bibitemstart{Minsky}
\bibinfo{author}{M.~Minsky} (\bibinfo{year}{1967}):
  \emph{\bibinfo{title}{Computation -- Finite and Infinite Machines}}.
\newblock \bibinfo{publisher}{Prentice Hall, Englewood Cliffs, NJ}.
\bibitemend

\bibitemstart{PaSa89}
\bibinfo{author}{Gh.~P\u aun} \& \bibinfo{author}{L.~S\^antean}
  (\bibinfo{year}{1989}): \emph{\bibinfo{title}{Parallel communicating grammar
  systems: The regular case}}.
\newblock {\sl \bibinfo{journal}{Ann. Univ. Bucharest, Ser. Matem.-Inform.}}
  \bibinfo{volume}{38}, pp. \bibinfo{pages}{55--63}.
\bibitemend

\bibitemstart{hb}
\bibinfo{editor}{G.~Rozenberg} \& \bibinfo{editor}{A.~Salomaa}, editors
  (\bibinfo{year}{1997}): \emph{\bibinfo{title}{Handbook of Formal Languages}}.
\newblock \bibinfo{publisher}{Springer-Verlag, Berlin}.
\bibitemend

\bibitemstart{Va}
\bibinfo{author}{Gy. Vaszil} (\bibinfo{year}{2007}):
  \emph{\bibinfo{title}{Non-returning PC grammar systems generate any
  recursively enumerable language with eight context-free components}}.
\newblock {\sl \bibinfo{journal}{Journal of Automata, Languages and
  Combinatorics}} \bibinfo{volume}{12}, pp. \bibinfo{pages}{307--316}.
\bibitemend

\bibliographyend
\end{thebibliography}

\end{document}